\documentclass[11pt,prd,aps,amsfonts,eqsecnum,superscriptaddress,nofootinbib,longbibliography,notitlepage]{revtex4-1}

\usepackage{lmodern}    
\usepackage{microtype}  
\usepackage[english]{babel}
\usepackage{setspace}

\usepackage{fancyhdr}

\usepackage{caption}
\usepackage{graphicx}
\usepackage{xcolor}
\usepackage{subcaption}
\usepackage{rotating}
\usepackage{amsmath, amssymb, graphics, amsthm, isomath}
\usepackage{physics}
\usepackage{verbatim}
\usepackage{tikz}
\usepackage{tikz-cd}
\usepackage{circuitikz}
\usetikzlibrary{patterns}
\usetikzlibrary{through,hobby}
\usetikzlibrary{decorations.markings}
\usetikzlibrary{fadings,shadings}
\usetikzlibrary{plotmarks}
\usetikzlibrary{positioning}
\usetikzlibrary{decorations,arrows}
\usetikzlibrary{decorations.pathreplacing}
\usetikzlibrary{chains, scopes, positioning, backgrounds, shapes, fit, shadows, calc, arrows.meta, decorations.pathreplacing}
\usepackage{blkarray}
\usepackage{svg}
\usepackage{mathrsfs}
\usepackage{algorithm}
\usepackage{algpseudocode}

\usepackage{xcolor}
\definecolor{mycitecolor}{rgb}{0.0, 0.45, 0.85}   

\usepackage[
  colorlinks=true,
  linkcolor=mycitecolor,
  citecolor=mycitecolor,
  urlcolor=mycitecolor,
  hyperindex=true,
  linktocpage=true
]{hyperref}

\usepackage[capitalise,compress]{cleveref}
\usepackage{tcolorbox}

\newcommand{\diff}{\mathrm{d}}

\newcommand{\edge}{\text{edge}}

\DeclareMathOperator{\arccosh}{arccosh}

\renewcommand\equiv{:=}
\renewcommand\epsilon{\varepsilon}

\newtheoremstyle{upright-hang}
  {6pt}   
  {6pt}   
  {\normalfont\hangindent=1.5em\hangafter=1\relax} 
  {0pt}   
  {\bfseries} 
  {.}     
  {0.5em} 
  {}      

\theoremstyle{upright-hang}
\newtheorem{thm}{Theorem}

\newtheorem{prop}[thm]{Proposition}

\renewcommand{\thesection}{\arabic{section}}
\renewcommand{\thesubsection}{\thesection.\arabic{subsection}}
\renewcommand{\thesubsubsection}{\thesubsection.\arabic{subsubsection}}

\makeatletter
\renewcommand{\p@subsection}{}
\renewcommand{\p@subsubsection}{}
\makeatother

\makeatletter
\fancypagestyle{plain}{%
  \fancyhf{}%
  \fancyfoot[C]{\thepage}%
}
\fancypagestyle{revtexfooter}{%
  \fancyhf{}%
  \fancyfoot[C]{\thepage}%
}
\makeatother
\AtBeginDocument{\setstretch{1.25}\pagestyle{revtexfooter}}
\makeatletter
\renewcommand\frontmatter@title@format{\LARGE\bfseries\centering\parskip\z@skip}

\def\frontmatter@affiliationfont{\normalfont\selectfont}
\newcommand{\AndySectionStar}[1]{%
  \par\addpenalty\@secpenalty
  \addvspace{3.5ex \@plus 1ex \@minus .2ex}%
  \noindent{\normalfont\Large\bfseries #1\par}%
  \nobreak\vspace{2.3ex \@plus .2ex}%
  \@afterheading}
\newcommand{\AndySubsectionStar}[1]{%
  \par\addpenalty\@secpenalty
  \addvspace{3.25ex \@plus 1ex \@minus .2ex}%
  \noindent{\normalfont\large\bfseries #1\par}%
  \nobreak\vspace{1.5ex \@plus .2ex}%
  \@afterheading}
\newcommand{\AndySubsubsectionStar}[1]{%
  \par\addpenalty\@secpenalty
  \addvspace{3.25ex \@plus 1ex \@minus .2ex}%
  \noindent{\normalfont\normalsize\bfseries #1\par}%
  \nobreak\vspace{1.5ex \@plus .2ex}%
  \@afterheading}
\renewcommand\@seccntformat[1]{\csname the#1\endcsname\quad}
\def\@hangfrom@section#1#2#3{\@hangfrom{#1#2}#3}
\def\@hangfroms@section#1#2{#1#2}
\renewcommand\section{\@ifstar{\AndySectionStar}{\@startsection {section}{1}{\z@}%
  {-3.5ex \@plus -1ex \@minus -.2ex}%
  {2.3ex \@plus .2ex}%
  {\normalfont\Large\bfseries}}}
\renewcommand\subsection{\@ifstar{\AndySubsectionStar}{\@startsection{subsection}{2}{-\parindent}%
  {-3.25ex\@plus -1ex \@minus -.2ex}%
  {1.5ex \@plus .2ex}%
  {\normalfont\large\bfseries}}}
\renewcommand\subsubsection{\@ifstar{\AndySubsubsectionStar}{\@startsection{subsubsection}{3}{-\parindent}%
  {-3.25ex\@plus -1ex \@minus -.2ex}%
  {1.5ex \@plus .2ex}%
  {\normalfont\normalsize\bfseries}}}
\renewcommand*\l@section[2]{%
  \ifnum \c@tocdepth >\z@
    \addpenalty\@secpenalty
    \addvspace{0.75em \@plus\p@}%
    \begingroup
      \parindent\z@
      \rightskip\@pnumwidth
      \parfillskip-\@pnumwidth
      \leavevmode\bfseries
      \setlength\@tempdima{2.3em}%
      \advance\leftskip\@tempdima
      \hskip-\leftskip
      #1\nobreak\hfil\nobreak\hb@xt@\@pnumwidth{\hss #2}\par
    \endgroup
  \fi}
\renewcommand*\l@subsection{\@dottedtocline{2}{1.5em}{3.2em}}
\renewcommand*\l@subsubsection{\@dottedtocline{3}{3.8em}{4.1em}}
\makeatother

\usepackage{mathtools}
\tikzstyle{densely dashed}= [dash pattern=on 4pt off 3pt]

\usepackage{dsfont}

\allowdisplaybreaks

\begin{document}

\title{Non-perturbative saturation of Krylov complexity, and its implications in quantum gravity}

\author{Andrew Lucas}
\email{andrew.j.lucas@colorado.edu}
\affiliation{Department of Physics and Center for Theory of Quantum Matter, University of Colorado, Boulder CO 80309, USA}

\author{Amit Vikram}
\affiliation{Department of Physics and Center for Theory of Quantum Matter, University of Colorado, Boulder CO 80309, USA}

\begin{abstract}
 We study the dynamics of Krylov state complexity in finite-dimensional quantum systems.  Orthogonal polynomials built on a discrete energy spectrum crowd away from regions where the density of states is low at large Krylov index.  The physical implication of this result is that, even using a minimalist Krylov state complexity which is defined over the entire spectrum, the Krylov complexity provably stops growing a little past the Heisenberg length in every low-energy state.  This has interesting implications for the proposal that Krylov complexity is related to the wormhole length in 2D quantum gravity.
\end{abstract}

\maketitle

\tableofcontents

\section{Introduction}
 The ``complexity = volume" conjecture \cite{Susskind:2014rva, StanfordSusskindComplexity} has received much recent attention \cite{ComplexityAnything, RecentComplexityReview} as a potential quantitative link between concepts from quantum information and quantum gravity.  In a nutshell, the idea is as follows:  there is a natural notion of quantum state complexity corresponding to the circuit depth necessary to prepare an $n$-qudit state starting from a product state.  This complexity $C$ is expected to grow linearly with time $t$ until the Heisenberg time $t_{\mathrm{H}}\sim \mathrm{e}^{S(E)}$; here $S(E)$ is the entropy of the quantum system at energy $E$.  At the Heisenberg time, the state becomes as complicated as possible given the effectively finite-dimensional Hilbert space seen at energy $E$, so complexity should saturate.   This behavior is qualitatively the same as the behavior of the volume of the interior of a two-sided black hole in gravity, initially prepared in a simple Hartle-Hawking state.  Semiclassical calculations reveal that this volume also grows linearly with time at times well before the Heisenberg time.   Although we do not yet have quantitative control over non-perturbatively large times $t\sim t_{\mathrm{H}}$ in quantum gravity, we know that the interior volume should also saturate, consistent with the complexity = volume conjecture.  Predicting the quantitative details of this saturation is then a non-trivial test for any theory of non-perturbative quantum gravity \cite{Iliesiu:2021ari, Iliesiu:2024cnh}.

Recently, there has been progress on comparing volume to a somewhat distinct notion of complexity known as Krylov complexity \cite{Parker:2018yvk,spreadcomplexity, KrylovReview, SonnerKrylovReview}.  Most concretely, \cite{Lin,Sonner1,Sonner2,Heller:2024ldz} have shown quantitative relationships between the semiclassical Hamiltonian of 2D quantum (JT) gravity and (low-energy) Hamiltonian dynamics between eigenstates of the Krylov complexity operator (which we review in Section \ref{sec:krylov}) in a specific microscopic model for 2D gravity known as the double-scaled Sachdev-Ye-Kitaev (DSSYK) model \cite{MaldacenaStanford, KitaevSuh, BerkoozSYK}.  The precise quantitative setting studied in \cite{Sonner1,Sonner2,Heller:2024ldz} uses a definition of Krylov complexity that captures the \emph{entire} spectrum of DSSYK -- not merely the low-energy sector.    A priori then, one might worry that this Krylov complexity might only saturate at the maximal Heisenberg time $\mathrm{e}^{S_{\text{max}}}$ associated with infinite temperature dynamics, rather than the Heisenberg time scale $\mathrm{e}^{S(E)}$, which depends on the energy $E$ of the initial state.  This concern is serious enough that other authors \cite{Kar:2021nbm,VijaySaturation} have proposed differing interpretations of ``Krylov complexity = volume" to ensure saturation at the Heisenberg time.

In this paper, we prove that the simple formulation of Krylov complexity used in \cite{Sonner1,Sonner2,Heller:2024ldz} actually saturates at time scales only a little bit larger than the Heisenberg time, despite the fact that it is sensitive to the entire UV Hilbert space!   Our results rely on mathematical bounds from the theory of orthogonal polynomials on finite sets \cite{baik}, developed in Section \ref{sec:polynomials}.  The application of these results to the main puzzle introduced above is then presented in Section \ref{sec:discussion}.  In DSSYK, we will find that the Krylov complexity saturates at values $\le  \mathrm{e}^{S(E)}  \log  E_{\text{max}}$, where $ E_{\max}$ is the bandwidth of the spectrum.  Our results hold for any low-energy state, including but not limited to the thermofield double state, which is mostly commonly considered in the quantum gravity literature.
 
 Combined with the quantitative matching between semiclassical gravity and early time Krylov complexity dynamics, our results suggest how Krylov state complexity may indeed be consistent with ``complexity = volume" even in a non-perturbative theory of quantum gravity.   There may also be quantitative relationships between Krylov state complexity and other notions of complexity, such as quantum ergodicity \cite{dynamicalqergodicity, JTreconstruction, dynamicalqentanglement}, which we recently argued is naturally related to a minimalist proposal for a length operator in non-perturbative JT gravity~\cite{JTreconstruction}.  We leave a detailed exploration of these questions to future work.

\section{Krylov state complexity}\label{sec:krylov}
In this section, we first review the proposal of \cite{Lin} for a candidate wormhole length operator in 2D quantum gravity, based on Krylov complexity \cite{Parker:2018yvk, spreadcomplexity}.   

\subsection{Basic facts about Krylov complexity}

Consider a discrete quantum mechanical system with a total Hilbert space dimension $d \sim \mathrm{e}^{S_{\text{max}}}$ which is non-perturbatively large, and Hamiltonian\footnote{We will use $z$, rather than $E$, for energy until Section \ref{sec:discussion}.  $E$ will refer to energy in the gravity theory, and it will be a rescaled version of $z$.} \begin{equation}
    H = \sum_{m=1}^d z_m |z_m\rangle\langle z_m|.
\end{equation} We assume the spectrum is non-degenerate for simplicity.  The energy levels are drawn from a specific realization of a random draw from a density of states \begin{equation}
    \rho(z) = \mathrm{e}^{S_0} \rho_0(z), \label{eq:rhoE}
\end{equation}
where $S_0$ is taken to be large (but finite), and $\rho_0(z)$ is a function with compact support: \begin{equation}
    \rho_0(z) = \mathrm{\Theta}(z_{\text{max}} - |z|) \mathrm{e}^{\delta S(z)}. \label{eq:rho0E}
\end{equation}
Define the seed Krylov state \begin{equation}
    |K_0\rangle := \frac{1}{\sqrt{d}}\sum_{m=1}^d |z_m\rangle . \label{eq:unbiasedKrylov0}
\end{equation}
Then we inductively define all subsequent Krylov states $|K_n\rangle$ for $1\le n\le d-1$ as follows:  \begin{equation}
    |K_{n+1}\rangle \propto \left(1 - \frac{|K_n\rangle \langle K_n|}{\langle K_n|K_n\rangle} -\frac{|K_{n-1}\rangle \langle K_{n-1}|}{\langle K_{n-1}|K_{n-1}\rangle}  \right) H|K_n\rangle .
\end{equation}
As is well-known from the theory of orthogonal polynomials and Lanczos coefficients, this prescription generates an orthonormal basis in which the Hamiltonian $H$ is tridiagonal:  $\langle K_j | H|K_k\rangle = 0$ if $|j-k|>1$.   Having found this basis, it is often useful to write it as: \begin{equation}
    H |K_n\rangle = a_n |K_{n+1}\rangle + b_n |K_n\rangle + a_{n-1}|K_{n-1}\rangle . \label{eq:krylovchain}
\end{equation}
By construction, the coefficients $a_n$ and $b_n$ are real.   The coefficients $a_n$ are called the Lanczos coefficients.

It has been argued \cite{Lin,Sonner1,Sonner2,Heller:2024ldz} that the Krylov state complexity operator \begin{equation}
    X := \sum_{n=0}^{d-1} n |K_n\rangle \langle K_n|  \label{eq:krylov_length}
\end{equation}
is a good candidate for a wormhole length operator in DSSYK.   We will return to the specific context of quantum gravity in Section \ref{sec:discussion}.  

For now, though, we focus on some generic results about Krylov complexity that follow simply from the asymptotics of the Krylov coefficients \begin{equation}
        a_m = C m^\alpha + \cdots , \label{eq:am_asymptotics}
    \end{equation}
    at large $m$.  In what follows, we assume that the $\cdots$ are not rapidly oscillating small fluctuations.  Suppose that we are handed a state \begin{equation}
    |\psi(t)\rangle = \sum_{m=0}^{d-1}\psi_m(t)|K_m\rangle.
\end{equation}Notice that  \begin{equation}
    \frac{\mathrm{d}}{\mathrm{d}t}\langle X\rangle = \langle \mathrm{i}[H,X]\rangle = \sum_{m=0}^{d-2} \mathrm{i}a_m \left[\bar \psi_m \psi_{m+1}-\bar\psi_{m+1}\psi_m \right].
    \end{equation}
    If $\psi_m$ is a ``smooth" function of $m$, we might heuristically estimate that, given \eqref{eq:am_asymptotics}, \begin{equation}
        \frac{\mathrm{d}}{\mathrm{d}t} \langle X\rangle \sim \langle X\rangle^\alpha.
    \end{equation}
    Supposing that $0<\alpha<1$, for example, we would then deduce that $\langle X\rangle \sim t^{1/(1-\alpha)}$.  In general, such a scaling is basically inevitable given \eqref{eq:am_asymptotics}.  Assuming that $b_n$ is approximately constant, which is reasonable when $\rho_0(E)$ is an even function, then for any wave packet with most support at large $m$, 
    \begin{equation}
        \frac{\mathrm{d}^2}{\mathrm{d}t^2}\langle X^\beta \rangle = -\langle [H,[H,X^\beta]]\rangle  \sim C^2 \sum_m \left(2\alpha |\psi_m|^2 + \beta(\beta-1) |\psi_m - \psi_{m-2}|^2  \right)m^{\beta-2+2\alpha}.
    \end{equation}
    Taking $\beta = 2(1-\alpha)$ we find that $\partial_t^2 \langle X^\beta \rangle \gtrsim 1$, insinuating that the growth $X \sim t^{1/(1-\alpha)}$ can't be avoided.

    This is an important observation because, as we explained, Krylov complexity is supposed to grow linearly with time, after a short initial transient:  $\langle X\rangle \approx vt$ until $t\gtrsim \mathrm{e}^{S(E)}$.  This statement should be true both about the mean \emph{and} typical values of $X$.  We deduce that $\alpha=0$ in \eqref{eq:am_asymptotics}.
 
Another important observation is as follows.  At relatively short time scales (compared to $\mathrm{e}^{S_0}$) where we see the linear complexity growth, we expect that the Krylov basis $|K_n\rangle$ is very well approximated by a discretization of the continuous orthogonal polynomials $p_n(z)$ generated by the weight function $\rho(z)$: namely, $p_n(z)$ is a polynomial of degree $n$ and \begin{equation}
    \int \mathrm{d}z \; \rho(z) p_n(z) p_m(z) = \delta_{nm}.
\end{equation}
Indeed, these orthogonal polynomials also obey a recursion relation analogous to \eqref{eq:krylovchain}: \begin{equation}
    z p_n(z) = \tilde a_n p_{n+1}(z) + \tilde b_n p_n(z) + \tilde a_{n-1}p_{n-1}(z). \label{eq:polynomialchain}
\end{equation}
It is known \cite{lubinsky} that the asymptotic behavior of the Lanczos coefficients $\tilde a_n$ defined in \eqref{eq:polynomialchain} are governed by the scaling of $\rho(z)$ at large $z$.  For any reasonable (not strongly oscillating) choice of $\log \rho(z)$, if \begin{equation}
    \log \rho(z) \sim -c_0 |z|^{1/\alpha} + \cdots 
\end{equation}
at large $|z|$, then $\tilde a_n$ have the asymptotic growth \eqref{eq:am_asymptotics}.  In order to see $\alpha=0$, the spectrum should therefore be bounded, precisely as we argued in \eqref{eq:rho0E}. 

The boundedness of the spectrum was also emphasized in \cite{Kar:2021nbm}; however,  \cite{Kar:2021nbm} chose a microcanonical Krylov seed rather than the choice of \cite{Lin,Sonner1,Sonner2,Heller:2024ldz}.  We will argue for a different resolution in this paper. 

\subsection{When does Krylov complexity stop growing?}

We have explained above that if Krylov complexity \eqref{eq:krylov_length} is a suitable candidate for a length operator in quantum gravity, the spectrum of the theory must be bounded.  We then trivially have the operator identity $\lVert X\rVert \le d-1$, so at sufficiently late times, complexity must stop growing.   However, in quantum gravity, the complexity is presumed to stop growing much faster, at the Heisenberg time scale $\mathrm{e}^{S(E)}$, where $E$ is the energy scale of the state.

This has a very profound implication for the Krylov complexity conjecture.  Let us take any initial state $|\psi_0\rangle$ of typical energy $E_0$ and write it in the energy eigenbasis: \begin{equation}
    |\psi_0\rangle = \sum_{n=0}^{d-1} c_n |z_n\rangle.
\end{equation}
We have the basic identity 
\begin{equation}
    \lim_{T\rightarrow \infty }\frac{1}{T}\int\limits_0^T \mathrm{d}t \; \langle \psi_0(t)|X|\psi_0(t) \rangle = \sum_{m=0}^{d-1}m \sum_{n=0}^{d-1} \left| \langle z_n | K_m\rangle \right|^2 |c_n|^2  .
\end{equation}
If the Krylov complexity saturates at $\mathrm{e}^{S(E_0)} \ll d$ in a generic low-energy state, evidently $|\langle z_n|K_m\rangle |$ must be very small once $m \gg \mathrm{e}^{S(E_0)}$.   

It is, at first glance, far from obvious that the Krylov basis defined above will have this property.  Indeed, $|K_0\rangle$ is sensitive to the entire energy spectrum.  The main point of this paper is that, with some small but necessary caveats, we do in fact have this property, for every state concentrated at typical energy $E_0$.

\section{Orthogonal polynomials with a discrete spectrum}
\label{sec:polynomials}
In this section, we will establish the desired claim:  that Krylov polynomials overly sample low-energy states at small Krylov index $m$, such that low-energy states are overwhelmingly supported on $m \lesssim \mathrm{e}^{S(E_0)}$.  This is based on a rigorous mathematical result about orthogonal polynomials built on discrete spectra.  Indeed,  in the literature it is known \cite{baik} that as $S_0\rightarrow \infty$, orthogonal polynomials ``saturate" in regions where $\rho(E)$ is sufficiently small.   Unfortunately, calculating when this saturation occurs is a non-trivial problem using the results of \cite{baik}; we will find it easier to simply establish the claim ourselves, using simpler tools from the theory of orthogonal polynomials. 

\subsection{Statement of the main result}
 Let $\mathcal{E} = \lbrace z_k\rbrace_{k\in\mathbb{Z}_d}$ denote the set of energy levels (with the $z_k$ being real-valued), with $|\mathcal{E}| = d$. The Krylov basis is a set of orthonormal polynomials $p_n(z)$ with respect to the measure
\begin{equation}
	\diff\mu(z) = \sum_{k =0}^{d-1} \delta(z-z_k) \diff z,
\end{equation}
satisfying $p_n(z) = \mathrm{\Theta}(z^n)$ and
\begin{equation}
	\int \diff\mu(z)\ p_n(z) p_m(z) = \delta_{mn}.
\end{equation}
Moreover, they satisfy the (diagonal) completeness relation at each $z_k$:
\begin{equation} \label{eq:pn_completeness}
    \sum_{n=0}^{d-1} \lvert p_n(z_k)\rvert^2 = 1.
\end{equation}
Note that this sums over all $d$ polynomials. Our main result, which is  proved in Appendix \ref{app:proof}, is that near the spectral edge, completeness is almost attained by a number of polynomials comparable to the size of an edge window:

\begin{thm}\label{thm:main}
    Select an edge interval $[z_0, z_{d_{\edge}-1}]$ containing $d_{\edge}$ levels, and an energy level $z_k$ within this interval ($0 < k < d_{\edge}$). Define the width ratio
    \begin{equation}
        \eta_k \equiv \frac{z_{d_{\edge}}-z_k}{z_{d-1}-z_{d_{\edge}}}
    \end{equation}
    and the logarithmic spectral rigidity potential at the edge:
    \begin{equation}
        \nu_{k,\edge} \equiv \frac{1}{d_{\edge}} \sum_{j \in \mathbb{Z}_{d_{\edge}} \setminus \lbrace k\rbrace} \ln\left(\frac{z_{d-1}-z_0}{|z_k-z_j|}\right).
    \label{eq:nuedgedef}
    \end{equation}
    Then, for any $\epsilon > 0$, as long as the edge is large enough such that:
    \begin{equation}
        \nu_{k, \edge} d_{\edge} \geq \ln\left\lbrace(\epsilon^2/d)^{1/2}\right\rbrace,
    \end{equation}
    the first $q_k$ polynomials attain completeness to resolution $1/(1+\epsilon^2)$ at $z_k$,
    \begin{equation}
        \sum_{n=0}^{q_k} |p_n(z_k)\rvert^2 \geq \frac{1}{1+\epsilon^2},
    \end{equation}
    for any $q_k$ satisfying:
    \begin{equation}
        q_k \geq d_{\edge}\left[\frac{\nu_{k,\edge} + \dfrac{1}{2 d_{\edge}}\ln\left(\dfrac{4d}{\epsilon^2}\right)}{2\sqrt{\eta_k}}\left(\dfrac{1+2\eta_k + 2\sqrt{\eta_k+\eta_k^2}}{\sqrt{1+\eta_k}+\sqrt{\eta_k}}\right) + 1-\frac{1}{d_{\edge}}\right].
	\label{eq:saturationbound}
    \end{equation}
\end{thm}

In Appendix~\ref{app:logpotential}, we prove that the logarithmic potential $\nu_{k,\edge}$ is self-averaging in a suitable thermodynamic limit $d_{\edge}\rightarrow \infty$, under assumptions on the spectrum that seem overwhelmingly reasonable for DSSYK and other generic models with reasonable spectral rigidity properties: 
\begin{equation}
   \lim_{d_{\edge} \to \infty} \nu_{k, \edge} =\dfrac{\displaystyle  \int\limits_{z_0}^{z_{d_{\edge}}}\diff z\ \rho_0(z) \ln\left(\frac{z_{d-1}-z_0}{|z_k-z|}\right)}{\displaystyle \int\limits_{z_0}^{z_{d_{\edge}}}\diff z\ \rho_0(z)}.
    \label{eq:logpotential_thermo}
\end{equation}
We recall the definition of $\rho_0(z)$ in \eqref{eq:rhoE}. 

Returning to Theorem~\ref{thm:main}, note that if we take $d_{\edge}$ to be too close to $k$, $\eta_k$ is small and the first term becomes large. One reasonable choice is, say, $z_{d_{\edge}}-z_0 \approx (C+1) (z_k-z_0)$ where $C$ is some $\mathrm{\Theta}(1)$ constant.  In the $d_{\edge} \gg 1$ regime we have in mind, $z_{d_{\edge}}-z_0 \ll z_{d-1}-z_0$, and we can estimate (with the second $\eta_k$-dependent factor in parentheses in the first term in Eq.~\eqref{eq:saturationbound} being $\approx 1$, and $\nu_{k,\edge}$ dominating in the first factor) saturation by the index
\begin{equation}
q_k \lessapprox d_{\edge}\left[\frac{\nu_{k,\edge}}{2}\sqrt{\frac{z_{d-1}-z_0}{C(z_k-z_0)}} + 1\right].
\label{eq:saturationbound2}
\end{equation}


We emphasize that Theorem \ref{thm:main} is a completely general property of the Krylov basis, which can be applied to any quantum dynamical system, including those that are believed to be holographic duals to theories of quantum gravity.

\subsection{Intuitive arguments}
\label{sec:intuitive}
There is a simple and intuitive argument for Theorem \ref{thm:main} (the proof, however, relies on very different techniques). Suppose for simplicity that the spectrum $\rho(z)$ of energy levels is even:  $\rho(z)=\rho(-z)$, and that it extends all the way to $|z|=\Lambda$.  Define \begin{equation}
    \cos\theta := \frac{z}{\Lambda}.
\end{equation}  A classic result \cite{szego} from the theory of orthogonal polynomials states that, if we build the orthogonal polynomials $p_n(z)$ with respect to the \emph{continuous} weight $\rho_0(z)$, they will asymptote to (up to constants): \begin{equation}
    p_n(\theta) \rightarrow  \frac{\cos(n\theta + \gamma(\theta))}{ \sqrt{\rho_0(\Lambda \cos \theta) \Lambda \sin \theta}} \label{eq:pn_asymptotics}
\end{equation}
for some function $\gamma(\theta)$.  

As argued previously, we expect the Krylov $|K_n\rangle$ to look like simple discretizations \begin{equation}
    |K_n\rangle \approx \sum_{m=1}^d  p_n(z_m)|z_m\rangle \label{eq:krylov_continuous}
\end{equation}
when $n$ is sufficiently small.  At the same time, \eqref{eq:krylov_continuous} is inconsistent with \eqref{eq:pn_completeness} and \eqref{eq:pn_asymptotics} as $n$ gets too large.  Suppose for example that \eqref{eq:pn_asymptotics} was an accurate estimate:  $\langle z_m|K_n\rangle \approx p_n(z_m)$ for all $n \le N$.   Then we find that the completeness relation (with completeness magnitude $e^{S_0}$) forces: \begin{equation}
   e^{S_0} \gtrsim  \frac{N}{2} \frac{1}{\rho_0(z_m) \Lambda \sin\theta_m }.
\end{equation}
This becomes inconsistent once \begin{equation}
    N \gtrsim N_* =  e^{S_0} \rho_0(z_m) \sqrt{\Lambda^2-z_m^2}.
\end{equation}

The most natural resolution to this tension is that the Krylov polynomials simply ``crowd out" from the edges once $n>N_*$.  Another justification for this is that once the wavelength $\mathrm{\Delta}\theta \sim n^{-1}$ is comparable to the mean energy spacing:  $n \sim \rho(z_m)$ (up to an additional factor of $\mathrm{d}\theta/\mathrm{d}z$), the added polynomials have already resolved all possible functions defined only at the discrete energies $z_k$ for $k\lesssim m$.   Therefore, the completeness relation ought to be satisfied by $n\sim N_*$.   

It is easy to see this crowding effect numerically.  Figure \ref{fig:CUE} shows the overlaps $\langle z_m|K_n\rangle^2$ for a random matrix drawn from the circular unitary ensemble (CUE) with $d=512$.  More concretely, after choosing the eigenphases randomly according to CUE, we rescale them to $[-1,1]$ and subsequently evaluate the Krylov $|K_n\rangle$.    Note that the (coarse-grained) density of states is uniform within the interval of support for CUE, so the crowding effect arises entirely from the $\sin\theta$ factor in \eqref{eq:pn_asymptotics}. The polynomials also show some peaking behavior near saturation, an effect discussed in ~\cite{DebarghyaThesis}.

\begin{figure}[t]
    \centering
    \includegraphics[width=0.75\linewidth]{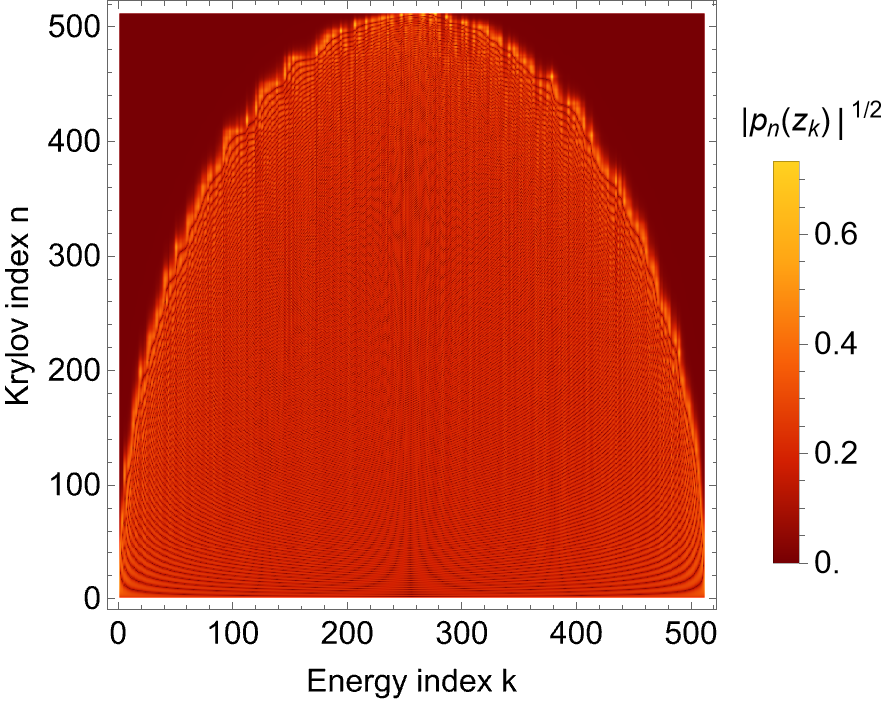}
    \caption{Demonstration of the ``withdrawal'' of the Krylov basis polynomials $p_n(z_k)$ from the edges for a CUE random matrix  with $d=512$ eigenphases (energy levels) rescaled to $[-1,1]$.} 
    \label{fig:CUE}
\end{figure}

\section{Implications for complexity = volume}
\label{sec:discussion}
Now that we have discussed Krylov complexity more quantitatively, we return to a discussion of quantum gravity.  We focus on the theory of quantum JT gravity \cite{Jackiw:1984je, Teitelboim:1983ux, Harlow:2018tqv} in two dimensions, which is particularly simple, and where the complexity = volume conjecture has been fairly quantitatively analyzed \cite{Lin,Sonner1,Sonner2,Heller:2024ldz}.  

In the semiclassical limit $S_0\rightarrow \infty$, quantum JT gravity is dual to a nonrelativistic quantum mechanics problem with a continuous and non-degenerate spectrum \cite{Harlow:2018tqv}: \begin{equation}
    H = \frac{p^2}{2} + 2\mathrm{e}^{-x}, \label{eq:semiclassical}
\end{equation}
where $[x,p]=\mathrm{i}$ and $x$ denotes the wormhole length operator.  This description is expected to break down at the Heisenberg time $t \sim \mathrm{e}^{S(E)}$ where $E$ is the energy of the state and $S$ is the entropy which, in JT gravity, takes the form
\begin{equation}
    S(E) = S_0 + 2\pi \sqrt{2E}.
\end{equation}
From the quantum gravity path integral, the breakdown of \eqref{eq:semiclassical} and its replacement with some discrete quantum theory is understood by resumming the higher-genus contributions to the path integral: see e.g. \cite{Iliesiu:2024cnh, Saad:2019lba}.   The quantum gravity path integral is performed in Euclidean time; unfortunately, it is very non-trivial to relate these Euclidean time calculations to the real-time saturation phenomenon of length.

The semiclassical effective Hamiltonian of JT gravity can be obtained from the Hilbert space of DSSYK as follows \cite{Lin, Sonner1, Sonner2}. We remind the reader that DSSYK~\cite{MaldacenaStanford, KitaevSuh, BerkoozSYK} corresponds to a theory of $N$ Majorana fermions (so the total Hilbert space dimension $d=2^{N/2}$), with all-to-all $q$-body interactions.  One introduces a parameter $\lambda = 2q^2/N$, and fixes $\lambda$ while $N\rightarrow \infty$.  The Krylov basis is defined as in \eqref{eq:unbiasedKrylov0}; in the semiclassical limit, as we discussed in Section \ref{sec:krylov}, this will be represented by the polynomial $p_0(z)$ in Sec.~\ref{sec:polynomials}. The rest of the Krylov basis is then represented by the higher polynomials $p_n(z)$. In this setting, JT gravity is obtained as a ``triple scaling'' $\lambda \to 0$ limit of DSSYK, with the identifications \cite{Lin,Sonner1,Sonner2}
\begin{align}
	x &= \lambda n, \label{eq:KrylovLength}\\
	z_k &= z_0+\lambda E_k
\end{align}
where $n$ is the Krylov index of $p_n(z)$, and $x$ and $E_k$ are respectively the wormhole length and the energy on the JT side.  We also note that the finite $\lambda$ DSSYK model has been related to sine-dilaton gravity \cite{SineDilatonDSSYK, Heller:2024ldz}.   Focusing for simplicity on $\lambda \rightarrow 0$, the Krylov chain Hamiltonian \eqref{eq:krylovchain} was shown to quantitatively reproduce \eqref{eq:semiclassical} at sufficiently small $x$ (well before the saturation scale). This limit takes the discrete Krylov index into a continuous length variable, and restricts the energy range of interest to $\mathrm{O}(\lambda)$ above the ground state. Another fact we will need is that the overall width of the DSSYK spectrum goes as $z_{d-1}-z_0 \sim 4J/\lambda$, making the JT region a small $\mathrm{O}(\lambda^2)$ fraction near the edge of the full DSSYK spectrum.

In some sense, the above mapping identifies JT gravity as a semiclassical/thermodynamic limit of a finite-dimensional quantum system (full DSSYK is presumed to be dual to a more general theory of 2D gravity). A key question is how the length behaves beyond the ``perturbative'' semiclassical regime. From general considerations, it is believed~\cite{Susskind:2014rva, Iliesiu:2021ari, Iliesiu:2024cnh} that a wormhole at energy $E$ would have a length that saturates after exploring the total number of states around energy $E$, roughly
\begin{equation}
    x_{\text{sat}} \sim \sqrt{2E} \mathrm{e}^{S(E)}.
    \label{eq:idealsaturation}
\end{equation}

For definiteness, let us call the prescription~\cite{Lin, Sonner1, Sonner2} in Eq.~\eqref{eq:KrylovLength} $x_K$. The Krylov basis forms an orthonormal basis of $d$ vectors for the Hilbert space. One may have therefore expected that the saturation of the Krylov basis happens only around $q \sim d$, the full size of the DSSYK spectrum, and therefore does not give a satisfactory Hilbert space size for length states. One proposal to address this problem~\cite{VijaySaturation} within the Krylov framework defines length in terms of the Krylov basis generated from a coherent Gibbs state at finite $\beta$ in place of Eq.~\eqref{eq:unbiasedKrylov0},
\begin{equation}
    \lvert \psi_\beta\rangle \propto \sum_{m=1}^{d} e^{-\beta E_m/2}\lvert z_m\rangle.
    \label{eq:finitetempKrylov0}
\end{equation}
Let us call the corresponding Krylov index $n_{\beta}$, and identify the length $x_{\beta} \propto n_{\beta}$ following \cite{VijaySaturation}. Then, Ref.~\cite{VijaySaturation} shows that the length $x_{\beta}(\lvert \psi_\beta(t)\rangle)$ \textit{of this same state}, evolved in time saturates at around $\exp[\mathrm{O}(S(E_{\beta}))]$, where~\cite{JTreconstruction} $E_{\beta} \sim \beta^{-2}$ is the energy around which this wavepacket is concentrated, with a spread of $\Delta E \sim \beta^{-3/2}$. As noted therein, this $\mathrm{e}^{\mathrm{O}(S(E))}$ saturation differs exponentially from the ideal $\mathrm{O}(\mathrm{e}^{S(E)})$ expectation, though is otherwise of a similar form and gives a finite Hilbert space corresponding to this length.

However, we note that if $\beta$ is finite (relative to the JT energies $E$), which is required for the finite Hilbert space, this prescription appears to have a dangerous incompatibility with the semiclassical description. In this limit, the spectrum becomes continuous, and one can define~\cite{Harlow:2018tqv, Iliesiu:2024cnh} an unambiguous semiclassical length operator $x_{\text{sc}}$ with eigenstates $\lvert x_{\text{sc}}\rangle$ in terms of the eigenstates $\lvert E\rangle$ of Eq.~\eqref{eq:semiclassical},
\begin{equation}
    \langle x_{\text{sc}}\vert E\rangle = 4\mathrm{K}_{\mathrm{i}\sqrt{8E}}(4 e^{-x_{\text{sc}}/2}),
\end{equation}
where $\mathrm{K}_{\nu}(z)$ is a modified Bessel function. By the Heisenberg uncertainty principle $\mathrm{\Delta}x_{\text{sc}}\mathrm{\Delta} p\gtrsim 1$, we heuristically have $\mathrm{\Delta}x_{\text{sc}} \gtrsim \sqrt{E}/\mathrm{\Delta}E \sim \beta^{1/2}$ for the state $\lvert \psi_{\beta}\rangle$ (for small $\beta$). This state has $\mathrm{\Delta} x_{\beta} = 0$ by definition, being its own Krylov seed, which differs substantially from the semiclassical estimate, making $x_{\beta} \neq x_{\text{sc}}$ in general. Two things are therefore apparent: (1) the $\beta \to 0$ limit is essential to correctly recover semiclassical JT dynamics, and (2) the saturation value of the Krylov length, both defined using and evaluated for the same finite temperature initial state \eqref{eq:finitetempKrylov0}, diverges in this limit with the accessible Hilbert space dimension $d_{\beta}$ at inverse temperature $\beta$.   A similar proposal \cite{Kar:2021nbm} uses microcanonical windows to define Krylov seeds, but similar concerns about Heisenberg uncertainty apply here too.


In this context, Sec.~\ref{sec:polynomials} implies that if the $\beta = 0$ Krylov procedure, which correctly recovers the semiclassical JT Hamiltonian~\cite{Lin, Sonner1, Sonner2}, is applied to a finite-energy initial state within the low-energy JT region of the spectrum (rather than the dynamics of the state $\lvert \psi_0\rangle$ itself, which is not contained within the JT region), then the length $x_K$ provably saturates fairly close to Eq.~\eqref{eq:idealsaturation}, with a mild logarithmic divergence that is determined by the parameter $\lambda$ rather than the overall Hilbert space dimension $d$.   To see this, we return to Eq.~\eqref{eq:saturationbound2} and take the spectral edge of size $d_{\edge}$ to correspond to the JT region $E_{\edge} = O(1)$ [or $z = z_0 + O(\lambda)$], with $z_k$ being any energy of interest within this region with $E_{\edge} = (C+1)E_k$. The corresponding value of $x = \lambda q(z_k)$ determines the length by which a state at energy $E_k$ saturates, i.e.
\begin{equation}
    x_{K, \text{sat}}(E_k) \lesssim \lambda d_{\edge}\left[\frac{\nu_{k,\edge}}{2}\sqrt{\frac{4J}{\lambda^2 C E_k}} + 1\right]
    \label{eq:saturationbound3}
\end{equation}

What remains is to estimate $\nu_{k,\edge}$. Combining $z_k-z_0 = \lambda E$ with $z_{d-1} - z_0 \sim 4 J/\lambda$, we can estimate that the typical term in Eq.~\eqref{eq:nuedgedef} scales as
\begin{equation}
    \ln\left(\frac{z_{d-1}-z_0}{|z_k-z_j|}\right) \sim \ln\left(\frac{4J/\lambda}{\lambda (E_k - E_j)}\right).
\end{equation}
From Eq.~\eqref{eq:logpotential_thermo}, we get
\begin{equation}
    \lim_{d_{\edge}\to\infty} \nu_{k,\edge} = \ln\left(\frac{4J}{\lambda^2}\right) - \frac{\int_{E_0}^{E_{\edge}} \diff E\ \rho(E) \ln|E_k-E|}{\int_{E_0}^{E_{\edge}} \diff E\ \rho(E)}.
\end{equation}
The second term is finite, and the first term dominates as $\lambda \to 0$. Therefore, denoting the second term by $L_{\edge}(E_k)$, we can write
\begin{equation}
	\nu_{k,\edge} \sim \ln \frac{4J}{\lambda^2}-L_{\edge}(E_k).
\end{equation}

Combining these ingredients, Eq.~\eqref{eq:saturationbound3} implies for the saturation length of the Krylov basis on the JT side (note that the second term goes to $0$ in the triple scaling limit, and is merely retained for completeness):
\begin{align}
	x_{K, \text{sat}}(E) &\lesssim \mathcal{N}((C+1)E)\left[\sqrt{\frac{J}{CE}}\left\lbrace \ln \frac{4J}{\lambda^2}-L_{\edge}(E)\right\rbrace+\lambda\right] \nonumber \\
    &\sim \sqrt{(C+1)E} \mathrm{e}^{S_0 + 2\pi \sqrt{2(C+1)E}} \left[\sqrt{\frac{J}{CE}}\left\lbrace \ln \frac{4J}{\lambda^2}-L_{\edge}(E)\right\rbrace+\lambda\right],
	\label{eq:KrylovSatLengthEstimate}
\end{align}
where $\mathcal{N}(E) \sim \sqrt{E} e^{S(E)}$ is the number of energy levels up to energy $E$ (the ``action'' variable in \cite{JTreconstruction}). We emphasize again that key differences between the saturation derived here and that in Ref.~\cite{VijaySaturation} include (1) the saturation value does not diverge with the inverse temperature $\beta$ (with $\beta = 0$ in our case), and (2) we work precisely within the $\beta = 0$ prescription of \cite{Lin, Sonner1, Sonner2} that gives the correct semiclassical JT Hamiltonian in the triple-scaled limit. We also note that this is \textit{independent} of the initial state whose dynamics we are interested in (different from the state used to generate the Krylov basis), and our result only constrains the length at which the component of any initial state at energy $E$ saturates. Additionally, the edge cutoff parameter $C$ can be chosen to be any finite number (including with $E$ dependence) in the thermodynamic limit, for example $C = J/E$; in that case, our bound on the saturation value becomes
\begin{align}
    x_{K, \text{sat}}(E) &\lesssim \sqrt{E+J} \mathrm{e}^{S_0 + 2\pi \sqrt{2(E+J)}} \left[\ln \frac{4J}{\lambda^2}-L_{\edge}(E)+\lambda\right] \nonumber \\
    &\approx \sqrt{E} \mathrm{e}^{S(E)} \left[\ln \frac{4J}{\lambda^2}-L_{\edge}(E)+\lambda\right] \text{ for } E \gg J.
\end{align}
For high energy states, this corresponds to length saturation at $\mathrm{O}(\sqrt{E} \mathrm{e}^{S(E)} \log (1/\lambda))$ which is comparable to $\mathrm{e}^{S(E)}$ up to the logarithmic factor.

Further, the completeness relation enforces $q \geq d_{\edge}-1$ before completeness is attained in \textit{all} points in the edge region (a basis cannot be complete in a subspace without having at least as many elements as the dimension of the subspace). For $x$, this typically translates to a lower bound corresponding to the second term in Eq.~\eqref{eq:KrylovSatLengthEstimate}:
\begin{equation}
	x_{K, \text{sat}}(E) \gtrsim \lambda \mathcal{N}((C+1)E). \label{eq:lower4}
\end{equation}
The naive saturating length is much closer to our upper bound \eqref{eq:KrylovSatLengthEstimate} than our lower bound \eqref{eq:lower4}.

These bounds should be contrasted with the ``ergodicity'' based prescription of \cite{JTreconstruction}, which constructs a length operator $x_{\theta}$ sensitive to the ``ergodic'' exploration of all states in a discrete Hilbert space corresponding to a given range of energies. There, the saturation length is shown to be
\begin{equation}
	x_{\theta, \text{sat}}(E) \sim \mathcal{N}(E) \sim \sqrt{E} e^{S_0 + 2\pi \sqrt{2E}},
\end{equation}
which directly corresponds to the size of the JT Hilbert space available up to energy $E$, matching the ideal expectation in Eq.~\eqref{eq:idealsaturation}. Eq.~\eqref{eq:KrylovSatLengthEstimate} constrains the saturation of the $\beta = 0$ Krylov basis to be very close to the ergodicity prescription, differing in order-of-magnitude from $O(\exp[S(E)])$ at most by a factor of $\log(1/\lambda)$ that mildly diverges in the $\lambda \to 0$ limit. We note that Eq.~\eqref{eq:KrylovSatLengthEstimate} is an upper bound, and the saturation may plausibly occur much closer to $\mathcal{N}(E)$ than what we can prove. We also regard this as additional evidence pointing towards a general nonperturbative relation between Krylov complexity and ergodicity in JT gravity, as conjectured in \cite{JTreconstruction}. 

A final relevant remark is that the DSSYK Hamiltonian need not exactly match \eqref{eq:semiclassical} at finite $\lambda$.  For example, the Hamiltonian of the dual gravity theory to DSSYK may be \cite{Heller:2024ldz} \begin{equation}
    H = \frac{p^2}{2} + 2\mathrm{e}^{-x} + \alpha  \lambda p^4 + \cdots .
\end{equation}
If one tracks the quantum evolution of a wave function in the Krylov basis of DSSYK vs. the ergodic basis of \cite{JTreconstruction} intended for pure JT, differences will arise at the same order as the $\alpha \lambda p^4$ term in the equation above.   Perhaps after accounting for these $\lambda$-dependent corrections to $H$, there is a ``simple" unitary that converts between the Krylov basis and the eigenbasis of the ergodic length operator from \cite{JTreconstruction}.

\section*{Acknowledgements}

This work was supported by the Heising-Simons Foundation under Grant 2024-4848. A.V. was supported by NIST, and by the National Science Foundation under Grant Number 1734006 (Physics Frontier Center).

GPT 5.4 Thinking pointed to Ref.~\cite{CVP} on the Christoffel Variational Principle, which was used to derive our main result, and to \eqref{eq:pn_asymptotics}. GPT 5.5 Thinking pointed to Ref.~\cite{ResidualPoly} on residual polynomials, while GPT 5.5 Pro pointed to \cite{baik}. Claude Fable 5 (and subsequently Opus 4.8) suggested improvements of an initially human-constructed variational ansatz and the resulting bound on being prompted to do so (in Appendix \ref{app:proof}), which are incorporated into our final results above (with typo corrections by Fable 5 and Opus 4.8). GPT 5.6 Sol pointed out a ``reversed'' inequality implication in the last part of Appendix~\ref{app:proof} in a previous version, which has now been corrected (with other typos). Appendix~\ref{app:logpotential} incorporates suggestions from GPT 5.5 and Claude Fable 5. All derivations were carried out and written up by the human authors, who take full responsibility for the claims.

\begin{appendix}

\makeatletter


\let\@sectioncntformat\@seccntformat

\let\@hangfrom@section\@hang@from


\renewcommand\section{\@ifstar{\AndySectionStar}{\@startsection{section}{1}{\z@}%

  {-3.5ex \@plus -1ex \@minus -.2ex}%

  {2.3ex \@plus .2ex}%

  {\normalfont\Large\bfseries}}}

\renewcommand\subsection{\@ifstar{\AndySubsectionStar}{\@startsection{subsection}{2}{\z@}%

  {-3.25ex\@plus -1ex \@minus -.2ex}%

  {1.5ex \@plus .2ex}%

  {\normalfont\large\bfseries}}}

\renewcommand\subsubsection{\@ifstar{\AndySubsubsectionStar}{\@startsection{subsubsection}{3}{\z@}%

  {-3.25ex\@plus -1ex \@minus -.2ex}%

  {1.5ex \@plus .2ex}%

  {\normalfont\normalsize\bfseries}}}

\makeatother

\renewcommand{\thesubsection}{\thesection.\arabic{subsection}}
\renewcommand{\theequation}{\thesection.\arabic{equation}}
\renewcommand{\thesubsubsection}{\thesubsection.\arabic{subsubsection}}
\makeatletter
\renewcommand\section{\@ifstar{\AndySectionStar}{\@startsection {section}{1}{-\parindent}%
  {-3.5ex \@plus -1ex \@minus -.2ex}%
  {2.3ex \@plus .2ex}%
  {\normalfont\Large\bfseries}}}
\renewcommand\subsection{\@ifstar{\AndySubsectionStar}{\@startsection{subsection}{2}{-\parindent}%
  {-3.25ex\@plus -1ex \@minus -.2ex}%
  {1.5ex \@plus .2ex}%
  {\normalfont\large\bfseries}}}
\renewcommand\subsubsection{\@ifstar{\AndySubsubsectionStar}{\@startsection{subsubsection}{3}{-\parindent}%
  {-3.25ex\@plus -1ex \@minus -.2ex}%
  {1.5ex \@plus .2ex}%
  {\normalfont\normalsize\bfseries}}}
\makeatother

\section{Proof of Theorem \ref{thm:main}}
\label{app:proof}

The central feature underlying our results is the nonperturbative withdrawal of the Krylov basis from the edges of the spectrum, such that the $\lvert \mathcal{K}_n\rangle$ become increasingly concentrated towards the center of the spectrum (of $H$) with increasing $n$ beyond a threshold. This behavior can be established using the Christoffel Variational Principle~\cite{CVP} as described below.  

The Christoffel-Darboux kernel for these polynomials is defined as
\begin{equation}
	\mathcal{K}_m(z,y) = \sum_{n=0}^{m} p_n(z) p_n^\ast(y).
\end{equation}
When $x=y$, the kernel is non-negative and measures the ``completeness'' of the first $m$ polynomials at $x$: $\mathcal{K}_m(x,x) \leq 1$, with equality indicating completeness.

The Christoffel Variational Principle (Theorem 9.2 in Ref.~\cite{CVP}) states that for all trial polynomials $Q_{m, z_0}(z)$ such that $Q_{m, z_0}(z_0) = 1$, $\deg(Q_{m, z_0}) \equiv m \leq q$,
\begin{align}
	\int\diff\mu(z) \lvert Q_{m, z_0}(z)\rvert^2 \geq \frac{1}{\mathcal{K}_q(z_0,z_0)} \label{eq:CVP}
\end{align}
Moreover, equality is achieved for $Q_{m, z_0}(z) = \mathcal{K}_q(z,z_0)/\mathcal{K}_q(z_0,z_0)$. While the explicit construction of such a polynomial is difficult, an efficient practical use of this principle is to make an ansatz for $Q_{m, z_0}(z)$ that allows us to limit the right hand side to as close to $1$ as possible. For definiteness, let us choose $z_0 = z_k$, a specific energy level ($k > 0$) in the spectrum.

To motivate our ansatz, let us take the order $q$ to be unrestricted. Then, a simple way to saturate \eqref{eq:CVP} subject to the constraint $Q_{m,z_k}(z_k) = 1$ is to choose $m=d-1$ and a polynomial that vanishes at all points in our spectrum:
\begin{equation}
	Q_{d-1,z_k}(z) = \prod_{j \in \mathbb{Z}_d \setminus \lbrace k\rbrace} \frac{z-z_j}{z_k-z_j}.
	\label{eq:exactzeroform}
\end{equation}
As $Q_{d-1,z_k}(z_{j\neq k}) = 0$, it follows that for $q=d-1$, $\mathcal{K}_{d-1}(z_k,z_k) = 1$. This is just the trivial statement that the first $d$ orthogonal Krylov vectors (recalling that $m$ starts at $0$) span the Hilbert space.

On the other hand, if we take $Q_{m,z_k}(z)$ to (broadly speaking) be a function that decays rapidly with $|x-z_k|$, such that
\begin{equation}
	|Q_{m,z_k}(z_{j\neq k})| \leq \frac{\epsilon}{\sqrt{d-1}},
	\label{eq:epsilonconstraint}
\end{equation}
where $0 < \epsilon \ll 1$, then for any $q \geq m$, \eqref{eq:CVP} gives
\begin{equation}
	\mathcal{K}_q(z_k,z_k) \geq \frac{1}{1+\epsilon^2},
\end{equation}
It is worth noting that while $\epsilon \ll 1$ is sufficient to ensure that the Krylov basis is mostly complete by $q=m$; if we set $\epsilon \ll 1/\sqrt{d}$, we can even show completeness short of only one basis vector.   

Intuitively, we would like to minimize $m$ to find the earliest index $q$ by which the Krylov basis is mostly complete at $z_k$. One option is to loosely take $Q_{m,z_k}(z)$ to decay rapidly with $|z-z_k|$. To ensure that $Q_{m,z_k}(z_{k\pm 1})$ is small as per \eqref{eq:epsilonconstraint}, we again need a very high order of polynomial to allow rapid fluctuations within a single level spacing. However, if we only require that this decay becomes relevant a certain finite distance from $z_k$ that skips over a fraction of points in the spectrum, then a low order polynomial suffices. To justify skipping over these points, we can use the exact zero form in \eqref{eq:exactzeroform} restricted to these points.


Let us now capture this intuition more precisely. We are primarily interested in $z_k$ near the lower edge of the spectrum (i.e. $z_0$). Take $\mathcal{E}_{\edge} \subseteq \mathcal{E}$ to be a window of $d_{\edge}$ consecutive states at (say) the lower edge of the spectrum, of width $z_{\edge}$, including $z_k$. We want our ``decay'' function to skip over all points in $\mathcal{E}_{\edge}$, which suggests the ansatz: 
\begin{equation}
	Q_{d_{\edge}-1+r, z_k}(z) = \left(\prod_{j \in \mathbb{Z}_{d_{\edge}} \setminus \lbrace k\rbrace} \frac{z-z_j}{z_k-z_j}\right) R_r(z),
	\label{eq:variationalansatz}
\end{equation}
where $R_r(z)$ is a rapidly decaying polynomial of degree $r$. Now our task is simply to get towards as small an $r$ as feasible such that \eqref{eq:epsilonconstraint} is satisfied.

More useful background comes from the theory of residual polynomials~\cite{ResidualPoly}. Where the Christoffel Variational Principle \eqref{eq:CVP} minimizes the $\mathrm{L}^2$ norm of a polynomial of degree at most $m$ given $Q_{m,z_0}(z_0) = 1$, the residual polynomial $R_{m,z_0}(z)$ of degree at most $m$ instead minimizes the $\sup$ norm $\sup_{z\in I} |R_{m,z_0}(z)|$ in an interval $I$, again given $R_{m,z_0}(z_0) = 1$.
In the connected interval $z \in [-1,1]$ that does not contain $z_0$, the residual polynomial is known to be a Chebyshev polynomial $T_m(z)$, defined by:
\begin{equation}
	T_m(z) = \cos(m \arccos(z)).
\end{equation}
For later use, we note that $|T_m(z \in [-1,1])| \leq 1$.

This suggests that Chebyshev polynomials are a class of rapidly decaying polynomials, and may be a useful component of a variational ansatz for our problem. In particular, we want to map $[z_{d_{\edge}},z_{d-1}]$ to $[-1,1]$ via
\begin{equation}
	\xi(z) = 2 \frac{z-z_{d_{\edge}}}{z_{d-1}-z_{d_{\edge}}}-1,
\end{equation}
and take $R_r(z)$ to be a Chebyshev polynomial in this domain:
\begin{equation}
	R_r(z) = \frac{T_r(\xi(z))}{T_r(\xi(z_k))} 
    \end{equation}
    such that 
    \begin{equation}|R_r(z \in [z_{d_{\edge}}, z_{d-1}])| \leq \frac{1}{|T_r(\xi(z_k))|} = \frac{1}{\cosh\left[r \arccosh\left(1-2\frac{z_k-z_{d_{\edge}}}{z_{d-1}-z_{d_{\edge}}}\right)\right]}.
\end{equation}

Now, returning to \eqref{eq:variationalansatz}, using $|z-z_j| < z_{d-1}-z_0$ we can bound the first factor according to:
\begin{equation}
	\left\lvert\prod_{j \in \mathbb{Z}_{d_{\edge}} \setminus \lbrace k\rbrace} \frac{z-z_j}{z_k-z_j}\right\rvert \leq e^{d_{\edge}\nu_{k,\edge}},
\end{equation}
where we have identified the spectral rigidity measure in the edge window, defined in \eqref{eq:nuedgedef}.  
We can therefore guarantee that \eqref{eq:epsilonconstraint} is satisfied if:
\begin{equation}
	\frac{e^{d_{\edge}\nu_{k,\edge}}}{\cosh\left[r \arccosh\left(1+2\frac{z_{d_{\edge}}-z_k}{z_{d-1}-z_{d_{\edge}}}\right)\right]} \leq \frac{\epsilon}{\sqrt{d}},
\end{equation}
which gives a lower bound on the degree $r$ of the Chebyshev polynomial:
\begin{equation}
	r \geq \frac{\arccosh\left[(d/\epsilon^2)^{1/2} e^{d_{\edge} \nu_{k,\edge}}\right]}{\arccosh\left(1+2\eta_k\right)}.
    \label{eq:rexactbound}
\end{equation}
where we have recognized
\begin{equation}
    \eta_k = \frac{z_{d_{\edge}}-z_k}{z_{d-1}-z_{d_{\edge}}}.
\end{equation}

To get a sense of what the above result implies, it is convenient to simplify this in terms of even more elementary functions, obtaining a less comprehensive but more intuitive bound than Eq.~\eqref{eq:rexactbound} (whose effect will turn out to be subleading in any case for large $d$, $d_{\\edge}$). From $\arccosh z = \ln(z+\sqrt{z^2-1})$, in the numerator we can use $\arccosh z \leq \ln (2z)$ for $z \geq 1$ (as $\sqrt{z^2-1} < z$ in this domain) and in the denominator $\ln(1+z) \geq z/(1+z)$ to get
\begin{equation}
    \frac{\arccosh\left[(d/\epsilon^2)^{1/2} e^{d_{\edge} \nu_{k,\edge}}\right]}{\arccosh\left(1+2\eta_k\right)} \leq \frac{d_{\edge} \nu_{k,\edge} + \dfrac{1}{2}\ln\left(\dfrac{4d}{\epsilon^2}\right)}{2\sqrt{\eta_k}\left(\dfrac{\sqrt{1+\eta_k}+\sqrt{\eta_k}}{1+2\eta_k + 2\sqrt{\eta_k+\eta_k^2}}\right)}
\end{equation}
for
\begin{equation}
    (d/\epsilon^2)^{1/2} e^{d_{\edge} \nu_{k,\edge}} \geq 1.
\end{equation}
Therefore, in this regime, Eq.~\eqref{eq:epsilonconstraint} is satisfied for all
\begin{equation}
    \frac{r}{d_{\edge}} \geq \frac{\nu_{k,\edge} + \dfrac{1}{2 d_{\edge}}\ln\left(\dfrac{4d}{\epsilon^2}\right)}{2\sqrt{\eta_k}}\left(\dfrac{1+2\eta_k + 2\sqrt{\eta_k+\eta_k^2}}{\sqrt{1+\eta_k}+\sqrt{\eta_k}}\right).
\end{equation}

This means that at a given $z_k$, near-completeness to any $\epsilon < 1$ is attained by the Krylov index $q(z_k) = d_{\edge}-1+r $ for any $r$ satisfying the bound above, i.e. we obtain \eqref{eq:saturationbound}.

\section{Continuum expression for the logarithmic potential}
\label{app:logpotential}

First, we have the following bound: 

\begin{prop} Let $\lbrace r_j\rbrace_{j=0}^{d_{\edge}-1}$ be a (not necessarily regular) lattice
of points in $[z_0, z_{d_{\edge}}]$ satisfying $r_{j+1} > r_j$, $|z_j-r_j| \leq \widetilde{\Delta_3}$ (evocative of the Dyson-Mehta spectral rigidity parameter~\cite{DysonMehtaSR}) and $\min_{j \neq k} |z_k-r_j| \geq \min_{k \neq j} |z_k - z_j| \equiv \delta z_{\min}$ (i.e. the lattice isn't too close to the single point $z_k$). Then, for any $\epsilon_J > \widetilde{\Delta_3}$, if $\mathcal{N}(z_k, \epsilon_J)$ is the number of $r_j$ within a distance of $\epsilon_J$ from $z_k$:
	\begin{equation}
		\mathcal{N}(z_k, \epsilon_J) \equiv \left\lvert \lbrace j \in \mathbb{Z}_{d_{\edge}}: |r_j - z_k| \leq \epsilon_J\rbrace\right\rvert,
	\end{equation}
	 we have
\begin{equation}
        \left\lvert \nu_{k,\edge} - \frac{1}{d_{\edge}} \sum_{j \in \mathbb{Z}_{d_{\edge}} \setminus \lbrace k\rbrace} \ln\left(\frac{z_{d-1}-z_0}{|z_k-r_j|}\right)\right\rvert \leq -\ln\left(1-\frac{\widetilde{\Delta_3}}{\epsilon_J}\right) + \frac{\mathcal{N}(z_k, \epsilon_J)}{d_{\edge}}\ln\left( \frac{z_{d_{\edge}}-z_0}{\delta z_{\min}}\right).
        \label{eq:rigidity_fluctuations}
    \end{equation}
\end{prop}
\begin{proof}
    We have, writing $z_j = r_j + \delta z_j$,
    \begin{equation}
        \ln \frac{1}{|z_k - z_j|} - \ln \frac{1}{|z_k-r_j|} = \ln\frac{|z_k-r_j|}{|z_k-z_j|} = -\ln\left|1-\frac{\delta z_j}{z_k-r_j}\right|.
        \label{eq:nu_diff_setup}
    \end{equation}
    Separate the $r_j$ into points no further than and further than $\epsilon_J$ from $z_k$:
    \begin{align}
        \mathcal{J}_1 &= \lbrace j: j\in \mathbb{Z}_{d_{\edge}}\setminus \lbrace k\rbrace , |r_j-z_k| > \epsilon_J\rbrace, \nonumber \\
        \mathcal{J}_2 &= \lbrace j: j\in \mathbb{Z}_{d_{\edge}}\setminus \lbrace k\rbrace , |r_j-z_k| \leq \epsilon_J\rbrace.
    \end{align}
    Then, for $j \in \mathcal{J}_1$, noting that $|\delta z_j| \leq \widetilde{\Delta_3}$, and recalling that $\epsilon_J > \widetilde{\Delta_3}$, we have
    \begin{equation}
        \left|-\ln\left|1-\frac{\delta z_j}{z_k-r_j}\right|\right| \leq -\ln \left(1-\frac{|\delta z_j|}{|z_k-r_j|}\right),
    \end{equation}
    with equality attained for $\delta z_j/(z_k-r_j) \geq 0$ and the strict inequality applying otherwise. Overall,
    \begin{equation}
        \left|\ln \frac{1}{|z_k - z_j|} - \ln \frac{1}{|z_k-r_j|}\right| \leq -\ln \left(1-\frac{|\delta z_j|}{|z_k-r_j|}\right)
    \end{equation}
    as $-\ln|1-x|$ is monotonically increasing with $x$ for $x \in [0,1)$, we have
    \begin{equation}
        -\ln \left(1-\frac{|\delta z_j|}{|z_k-r_j|}\right) \leq -\ln\left(1-\frac{\widetilde{\Delta_3}}{\epsilon_J}\right).
        \label{eq:J1bound}
    \end{equation}
    
    For $j \in \mathcal{J}_2$, we apply the following bound that applies for all $j \in \mathbb{Z}_{d_{\edge}}$. We note that there are two different cases: $|z_k - r_j| < |z_k-z_j|$ and vice-versa. We can combine both cases as
    \begin{equation}
    	\left|\ln\frac{|z_k-r_j|}{|z_k-z_j|}\right| = \ln \frac{|z_k-\zeta|}{|z_k-\xi|},
    \end{equation}
    where $(\zeta, \xi) = (r_j, z_j)$ if $|z_k - r_j| \geq |z_k-z_j|$ and $(\zeta, \xi) = (z_j, r_j)$ otherwise. Identifying $|z_k-\xi| \geq \delta z_{\min} \equiv \min_{j \neq k; j,k \in \mathbb{Z}_{d_{\edge}}}|z_j-z_k|$ (by assumption for $\xi = r_j$), and using $|z_k-\zeta| \leq |z_{d_{\edge}}-z_0|$ with the monotonically increasing nature of $\ln x$ for $x > 0$, we have
    \begin{equation}
    	\left|\ln\frac{|z_k-r_j|}{|z_k-z_j|}\right| \leq \ln\left( \frac{z_{d_{\edge}}-z_0}{\delta z_{\min}}\right).
    	\label{eq:J2bound}
    \end{equation}
    Combining Eqs.~\eqref{eq:J1bound} and \eqref{eq:J2bound} with Eq.~\eqref{eq:nu_diff_setup} (and implicitly, $|\mathcal{J}_1| \leq d_{\edge}$), we get
    \begin{equation}
    	\left\lvert \nu_{k,\edge} - \frac{1}{d_{\edge}} \sum_{j \in \mathbb{Z}_{d_{\edge}} \setminus \lbrace k\rbrace} \ln\left(\frac{z_{d-1}-z_0}{|z_k-r_j|}\right)\right\rvert \leq -\ln\left(1-\frac{\widetilde{\Delta_3}}{\epsilon_J}\right) + \frac{|\mathcal{J}_2|}{d_{\edge}}\ln\left( \frac{z_{d_{\edge}}-z_0}{\delta z_{\min}}\right).
    \end{equation}
    Noting that $|\mathcal{J}_2| \leq \mathcal{N}(z_k, \epsilon_J)$, we get Eq.~\eqref{eq:rigidity_fluctuations}.
\end{proof}

While we have made an assumption that the lattice of points does not coincide too closely with $z_k$, this can easily be arranged by shifting the lattice by $\delta z_{\min}$ away from $z_k$ if required, changing $\widetilde{\Delta_3}$ by only an amount comparable to $\delta z_{\min}$. In the continuum limit, assuming that the $r_j$ form an approximately regularly spaced lattice within windows of approximately uniform density of states such that (1) they define a convergent Riemann sum in this limit, and (2) it is reasonable to model their behavior by lattices measuring the ``rigidity'' of the eigenvalues of standard random matrix ensembles\footnote{This is implicitly an assumption about how the thermodynamic limit $d\to\infty$ and $d_{\edge} \to \infty$ is taken in a given system, which we expect to be quite generically true for models such as DSSYK.}, such as the circular unitary ensemble~\cite{Mehta}, we get the following simplifications. 

First, let us allow $\widetilde{\Delta_3}$, $\delta z_{\min}$, and $\epsilon_J$ to depend on $d_{\edge}$ when taking the continuum limit. By the convergence of the Riemann sum, we obtain from Eq.~\eqref{eq:rigidity_fluctuations}
\begin{align}
    &\left\lvert \lim_{d_{\edge} \to \infty} \nu_{k,\edge} - \int_{z_0}^{z_{d_{\edge}}}\diff z\ \widetilde{\rho}(z) \ln\left(\frac{z_{d-1}-z_0}{|z_k-z|}\right)\right\rvert  \nonumber \\
    &\leq \lim_{d_{\edge} \to \infty} \left[-\ln\left(1-\frac{\widetilde{\Delta_3}(d_{\edge})}{\epsilon_J(d_{\edge})}\right) + \frac{\mathcal{N}_{d_{\edge}}(z_k, \epsilon_J(d_{\edge}))}{d_{\edge}}\ln\left( \frac{z_{d_{\edge}}-z_0}{\delta z_{\min}(d_{\edge})}\right)\right].
    \label{eq:nuedgecontinuumbound}
\end{align}
where $\widetilde{\rho}(z) \geq 0$ is the density of states (i.e. $\widetilde{\rho}(z) \propto \mu(z)$), assumed to be well-behaved at $z_k$ and normalized so that
\begin{equation}
    \int\limits_{z_0}^{z_{d_{\edge}}}\diff z\ \widetilde{\rho}(z) = 1.
\end{equation}

We will now make the following weak assumptions about spectral statistics. The first is that $\delta z_{\min}(d_{\edge})/(z_{d_{\edge}}-z_0) = \mathrm{\Omega}(d_{\edge}^{-a})$ for $a > 0$, the second is that $\widetilde{\Delta_3}(d_{\edge})/(z_{d_{\edge}}-z_0) = \mathrm{O}(d_{\edge}^{-\epsilon_3})$ for $\epsilon_3 \in (0,1)$, and the third is:
\begin{equation}
	\mathcal{N}_{d_{\edge}}(z_k, \epsilon_J(d_{\edge})) = \mathrm{O}\left(d_{\edge} \frac{\epsilon_J(d_{\edge})}{z_{d_{\edge}} - z_0}\right),
\end{equation}
which is a direct statement of spectral rigidity (the number of levels within an interval is proportional to the size of the interval,  in order of magnitude). For example, these are rigorously known to be satisfied in CUE, which has with high probability (1) $\delta z_{\min} = \Omega(d^{-4/3-\epsilon_{\min}})$ for $\epsilon_{\min} > 0$ \cite{BenArousBourgade}, (2) $\widetilde{\Delta_3} = O(\ln d)/d$~\cite{ArguinBeliusBourgade}, and (3) the number of levels in an interval being proportional to the length of the interval to leading order is a direct consequence of spectral rigidity~\cite{Mehta}.

Then we get:
	\begin{align}
		&\left\lvert \lim_{d_{\edge} \to \infty} \nu_{k,\edge} - \int_{z_0}^{z_{d_{\edge}}}\diff z\ \widetilde{\rho}(z) \ln\left(\frac{z_{d-1}-z_0}{|z_k-z|}\right)\right\rvert  \nonumber \\
		&\leq \lim_{d_{\edge} \to \infty} \left[-\ln\left(1-\mathrm{O}(d_{\edge}^{-\epsilon_3})\frac{z_{d_{\edge}} - z_0}{\epsilon_J(d_{\edge})}\right) + \mathrm{O}\left( \frac{\epsilon_J(d_{\edge})}{z_{d_{\edge}} - z_0}\right)\ln d_{\edge}\right].
	\end{align}

Therefore, if we choose $\epsilon_J(d_{\edge})$ to scale such that e.g.
\begin{equation}
	\frac{\epsilon_J(d_{\edge})}{z_{d_{\edge}} - z_0} = c_{\epsilon} (\ln d_{\edge})^{-2}
\end{equation}
for some constant, $c_{\epsilon} > 0$ we get the scaling behavior
\begin{align}
	&\left\lvert \lim_{d_{\edge} \to \infty} \nu_{k,\edge} - \int_{z_0}^{z_{d_{\edge}}}\diff z\ \widetilde{\rho}(z) \ln\left(\frac{z_{d-1}-z_0}{|z_k-z|}\right)\right\rvert  \nonumber \\
	&\leq \lim_{d_{\edge} \to \infty} \left[\mathrm{O}(d_{\edge}^{-\epsilon_3} (\ln d_{\edge})^{2}) + \mathrm{O}\left((\ln d_{\edge})^{-1}\right)\right] = 0.
\end{align}

Therefore, in the continuum limit, we have 
\begin{equation}
    \lim_{d_{\edge} \to \infty} \nu_{k, \edge} = \int_{z_0}^{z_{d_{\edge}}}\diff z\ \widetilde{\rho}(z) \ln\left(\frac{z_{d-1}-z_0}{|z_k-z|}\right),
\end{equation}
recovering Eq.~\eqref{eq:logpotential_thermo}. 

\end{appendix}

\vspace{-1em}   

\phantomsection
\section*{References}
\addcontentsline{toc}{section}{References}
\makeatletter
\let\bibsection\relax
\makeatother
\bibliography{refs}
\end{document}